\documentclass[onefignum,onetabnum]{siamart190516}
\usepackage{algorithmic}
\usepackage{graphicx}
\graphicspath{./}

\title{Necessary and Sufficient Condition for Satisfiability of a Boolean Formula in CNF and its Implications on P versus NP problem}

\author{Manoj Kumar\thanks{Assistant Professor in Computer Application, Government College of Teacher Education, Dharamshala, Dist. Kangra, Himachal Pradesh, India.
		(\email{manoj6689@gmail.com})}}
	
\begin{document}
\maketitle

\begin{abstract}
	Boolean satisfiability problem has applications in various fields. An efficient algorithm to solve satisfiability problem can be used to solve many other problems efficiently. The input of satisfiability problem is a finite set of clauses. In this paper, properties of clauses have been studied. A type of clauses have been defined, called fully populated clauses, which contains each variable exactly once. A relationship between two unequal fully populated clauses has been defined, called sibling clauses. It has been found that, if one fully populated clause is false, for a truth assignment, then all it's sibling clauses will be true for the same truth assignment. Which leads to the necessary and sufficient condition for satisfiability of a boolean formula, in CNF. The necessary and sufficient condition has been used to develop a novel algorithm to solve boolean satisfiability problem in polynomial time, which implies, P equals NP. Further, some optimisations have been provided that can be integrated with the algorithm for better performance. 
	
\end{abstract}



\section{Introduction}
Boolean Satisfiability problem is a NP-complete problem.\cite{Karp1972} It implies, all other NP-complete problems can be reduced to Boolean Satisfiability problem. So, if there exist an algorithm that can solve Boolean Satisfiability problem in polynomial time, then every other NP-complete problem can be solved in polynomial time. It has lead to formulation of P versus NP problem defined by Stephen Cook in \cite{Cook2006}. History and importance of P versus NP problem has been discussed in detail in \cite{Cook2006}.

In August, 2020, in \cite{Bearden2020}, an efficient dynamical-system approach to solve Boolean satisfiability problems have been presented. But also, in \cite{Bearden2020}, it has been stated, "Although these analytical and numerical results do not settle the famous P vs. NP question, they show that appropriately designed physical systems are very useful tools for new avenues of research in constraint satisfaction problems.", which implies, non-existence of an polynomial time algorithm to solve boolean satisfiability problem, till August, 2020. 

In this paper, properties of clauses have been studied, novel relationships have been defined among clauses, and a necessary and sufficient condition has been established that determine satisfiability of any boolean formula in CNF, which has been used to develop a polynomial time algorithm to solve boolean satisfiability problem. Which implies, that, boolean satisfiability problem belongs to $P$ class, which implies, $P=NP$. \cite{Karp1972}

\section{Boolean Satisfiability Problem}
As defined in \cite{Karp1972}, 
For the given clauses $C_1, C_2, C_3 \dots C_p$, we need to find whether conjuction of the given clauses is satisfiable or not. 

\section{Terminology used}
The terms literal, boolean variable, clause are used with same meaning as defined in \cite{Heule2015}.
A boolean formula in CNF is a conjuction of clauses. It can be represented as a finite set of clauses. \cite{Heule2015}
A set of boolean variables is called a variable set.

\section{Notations used}
The notations, representing basic relations between sets have been used as defined in \cite{jech2013set}. 

\subsection{Variable cases}
Let a variable, X, can be assigned values $x_1$, $x_2$ and $x_3$, independently,
then, it is written as:
$$
X=\begin{cases}
x_1 \\
x_2 \\
x_3 \\
\end{cases}
$$
\section{Tautology Clause}
A clause, which evaluates to $true$ for every valuation, is called a tautology clause. If a clause contains a complemented pair of literals, it is a tautology.\cite{Heule2015} In other words, If $\exists x \in T | \neg x \in T$, then $T$ is a taulogy clause.

\subsection{Significance of tautology clause in satisfiability problem}
\label{sec:st}
As a tautology clause always evaluates to $true$, that is represented by 1 in boolean algebra. Let F is a boolean formula in CNF, which containins a tautology clause, we can write 
\begin{equation}
F=C_1 \wedge T 
\end{equation}

where $T$ is a tautology clause.
$$
\implies F=C_1 \wedge 1 
$$
by using Identity property of Boolean algebra,
\begin{equation}
\implies F=C_1 
\end{equation}

Hence, The tautology clause has no effect on satisfiablity of a boolean formula, So, It can be ignored while solving satisfiability problem.

\subsection{Non-Tautology Clause}
A clause which is not a tautology is called a non-tautology clause.

\begin{lemma}
If $N$ is a non-tautology clause, then 
$\forall x \in N \implies \neg x \notin N$.	
\end{lemma}
\begin{proof}
	Given that, $N$ is not a tautology clause.
	Let, for the sake of contradiction, 
	$$\exists x \in N | \neg x \in N$$
	$\implies$ N is a tautology clause(from definition), which is not true. So, our assumption is wrong. Hence, 
	$$\forall x \in N \implies \neg x \notin N$$
	
\end{proof}

\begin{lemma}
	\label{lemma:x1c1}
	If $C$ is a clause, with n literals, such that, $\exists x_i \in C \; | \; x_i=1(true)$ then $C=1(true)$
\end{lemma}
\begin{proof}
	Given that, $C$ is a clause, by definition, $C$ is a disjunction of literals, so, we can write, 
	$$C=(x_1 \vee x_2 \vee \dots \vee x_i \vee \dots \vee  x_n)$$
	also, given that, $x_i=1$, so we can write,
	$$C=(x_1 \vee x_2 \vee \dots \vee 1 \vee \dots \vee  x_n)$$
	by using dominance law of boolean algebra,
	$$\implies C=1$$
	Hence proved.
\end{proof}

\begin{lemma}
	\label{lemma:c0x0}
	If $C$ is a clause, with n literals, such that, $C=0(false)$ then 
	$$x_i=0 \forall x_i \in C$$
\end{lemma}
\begin{proof}
	Given that, $C$ is a clause, by definition, $C$ is a disjunction of literals, so, we can write, 
	$$C=(x_1 \vee x_2 \vee \dots \vee  x_n)$$
	also, given that, 
	$$C=0(false)$$
	
	Suppose, for the sake of contradiction, $x_i=1$ for some $x_i \in C$. Using \cref{lemma:x1c1}, we get,
	$$C=1$$
	which is a contradiction, 
	Therefore, the assumption, $x_i=1$ is not true, hence, 
	\begin{equation}
	x_i=0 \; \forall x_i \in C
	\end{equation}
	Hence proved.
\end{proof}

\begin{theorem}
	\label{thm:c0d0}
	If $C$ and $D$ are clauses,
	such that, 
	$$	D \subseteq C $$
	and 
	$$ 	C=0 $$
	then $$D=0$$
\end{theorem}
\begin{proof}
	Given that,
	$$
	D \subseteq C
	$$
	and 
	\begin{equation}
	\label{eq:c01}
	C=0
	\end{equation} 
	
	using \cref{lemma:c0x0},
	\begin{equation}
	\label{eq:x0}
	x_i=0 \; \forall x_i \in C
	\end{equation}
	
	As $D \subseteq C$, 
	\begin{equation}
	\label{eq:xinD}
	\forall x \in D \implies x \in C
	\end{equation}
	From \cref{eq:x0} and \cref{eq:xinD}
	we have,
	$$
	x_j=0 \; \forall x_j \in D 
	$$
	We can write, 
	$$
	D=(0\vee0\vee0\vee \dots \vee 0)
	$$
	
	\begin{equation}
	\implies
	D=0
	\end{equation}
	Hence proved.
\end{proof}

\section{Clause over a variable set}
A non-tautology clause, $C$, is called a clause over a variable set, $V$, if,
$$(\forall x) (x\in C \; or \; \neg x \in C \implies x \in V)$$
For e.g. clauses, $C_1=\{x_1,x_2\}$ and $C_2=\{x_1,x_2,x_3\}$ are clauses over variable set, $V=\{x_1,x_2,x_3\}$
\subsection{Fully Populated Clause over a variable set}
A non-tautology clause, $C_{full}$, is called a fully populated clause over a variable set, $V$, if 
$$(\forall x) (x \in V \Leftrightarrow x\in C_{full} \; or\; \neg x \in C_{full})$$
For e.g. clause $C=\{x_1,\neg x_2\}$ is a fully populated clause over variable set, $V=\{x_1,x_2\}$

\begin{lemma}
	\label{lemma:cv}
	If $C$ is a clause over a variable set, $V$, then, $\exists V_{sub} \subseteq V$, such that, $C$ is a fully populated clause over $V_{sub}$.
\end{lemma}
\begin{proof}
	Given that, $C$ is a clause over a variable set, $V$, from definition,
	\begin{equation}
	\label{eq:fpc1}
		\implies (\forall x)( x \in C \; or \; \neg x \in C \implies x \in V )
	\end{equation}
	We define a variable set,
	\begin{equation}
		\label{eq:fpc2}
		V_{sub}=\{x|x \in C \; or \; \neg x \in C\}
	\end{equation}
	 $$\implies (\forall x) (x \in V_{sub} \Leftrightarrow x\in C \; or\; \neg x \in C)$$
	 also, from \cref{eq:fpc1} and \cref{eq:fpc2} 
	 $$\forall x\in V_{sub} \implies x\in V$$
	 $$\implies V_{sub} \subseteq V$$
	 Hence, $\exists V_{sub} \subseteq V$, such that, $C$ is a fully populated clause over variable set, $V_{sub}$.
\end{proof}
\begin{lemma}
	\label{lemma:fpcl2}
	If $C$ is a fully populated clause over a variable set, $V$, then $\forall D\subseteq C$, $\exists V_{sub} \subseteq V$, such that, $D$ is a fully populated clause over variable set $V_{sub}$
\end{lemma}
\begin{proof}
	Given that, $C$ is a fully populated clause over a variable set, $V$
	$$\implies (\forall x) (x \in V \Leftrightarrow x\in C \; or\; \neg x \in C)$$
	Suppose, $D\subseteq C$
	$$\implies \forall x \in D \implies x \in C$$
	$$\implies (\forall x) ( x\in D \; or \; \neg x \in D \implies x \in V)$$
	$\implies$ $D$ is a clause over $V$.
	
	From \cref{lemma:cv}, 
	$\exists V_{sub} \subseteq V$, such that, $D$ is a fully populated clause over $V_{sub}$
	
	Hence, $\forall D\subseteq C$, $\exists V_{sub} \subseteq V$, such that, $D$ is a fully populated clause over $V_{sub}$ 
	
\end{proof}

\begin{theorem}
		\label{lemma:ck0}
		For any given valuation to a variable set, $V$, there exist a fully populated clause, say $C_k$, over $V$, such that, 
		$$C_k=0(false)$$
\end{theorem}
\begin{proof}
	Let, the variable set, $V$, is given by,
	$$V=\{x_1,x_2,\dots, x_n \}$$
	where
	$$x_i= \begin{cases}
			0\\
			1\\
	\end{cases} \forall x_i \in V$$
	Let, each $x_i$ has been assigned any of the values given above.
	
	Now, we define a clause, $C_k$, depending upon the valuation assigned above, 
	
	$$C_k=\{y|y=\begin{cases}
				\neg x_i &if \; x_i=1\\
				x_i & if\;  x_i=0\\
	\end{cases} \; \forall x_i \in V \}$$
	
	$$\implies (\forall x) (x\in V \Leftrightarrow x \in C_k \;or\; \neg x \in C_k)$$
	
	$\implies$ $C_k$ is a fully populated clause. and,
	
	By putting values assigned for variables in $V$, in $C_k$, we get,
	$$x_i=0 \; \; \forall x_i \in C_k$$
	
	$$\implies C_k=0(false)$$
	Hence, for any given valuation to the variable set, $V$, there exists a fully populated clause, $C_k$, such that, $C_k=0(false)$
\end{proof}

\begin{theorem}
		\label{thm:2nclauses}
		For a given set of variables, $V$, with n variables, there exist $2^n$ fully populated clauses.
\end{theorem}
\begin{proof}
		For a given set of variables, $V$, with n variables, we can write a fully populated clause in the general form, given by,
		$$C=\{ x_1, x_2, \dots , x_n \}$$
		where 
		$$x_i=\begin{cases}
				x_i\\
				\neg x_i\\
		\end{cases}$$
		i.e. each $x_i$ can be assigned a value in two ways, independently. As there are $n$ number of variables, in a clause. So, by using basic principle of counting, there are $2^n$ ways, in which, a clause $C$ can be selected. Hence, for a given variable set, $V$, with $n$ variables, there exist $2^n$ number of fully populated clauses.
\end{proof}

\section{Sibling Clause}
Two unequal fully populated clauses over a common variable set, $V$, are called sibling clauses. In other words, If $A$ and $B$ are two non-tautology clauses, such that:
\begin{enumerate}
	\item $(\forall x) (x \in V \Leftrightarrow x\in A \; or \; \neg x \in A \Leftrightarrow x \in B \; or \; \neg x \in B)$
	\item  $\exists x \in A | \neg x \in B$
\end{enumerate}
then $A$ is a sibling clause of $B$ and vice-versa. For e.g. $\{\neg x_1,x_2\}$ and $\{x_1,\neg x_2\}$ are sibling clauses over a variable set, $\{x_1,x_2\}$.

\begin{theorem}
	\label{lemma:a0b1}
	If A and B are two sibling clauses, and $A=0(false)$, then $B=1(true)$.
\end{theorem}
\begin{proof}
	Given that $A$ and $B$ are sibling clauses.
	\begin{equation}
		\label{eq:s1}
		\implies \exists x_i \in A | \neg x_i \in B
	\end{equation}
	
	also, given that, $$A=0$$
	Using \cref{lemma:c0x0},
	$$\implies x_i=0 \; \forall x_i \in A$$
	from equation \cref{eq:s1},
	$$\implies \exists \neg x_i \in B | x_i=0$$
	put $y=\neg x_i$
	$$\implies \exists y \in B | y=1$$
	from \cref{lemma:x1c1},
	$$ B=1$$ Hence proved.
\end{proof}

\begin{theorem}
	\label{lemma:desibling}
	If $C_i$ and $C_j$ are two sibling clauses, over a variable set, $V$, and $\mathcal{P}(C_i)$ and $\mathcal{P}(C_j)$ are power sets of $C_i$ and $C_j$, respectively, then 
	$$\forall D \in \mathcal{P}(C_i) \; | \; D\notin \mathcal{P}(C_j)$$
	$$\exists E \in \mathcal{P}(C_j)$$
	such that, $D$ and $E$ are sibling clauses.
\end{theorem}
\begin{proof}
	Given that, $C_i$ and $C_j$ are sibling clauses over a variable set, $V$. It implies, by definition of sibling clauses, $C_i$ and $C_j$ are fully populated clauses over $V$.
	
	From \cref{lemma:fpcl2}, 
	\begin{equation}
		\forall D\subseteq C_i,\; \exists V_{sub} \subseteq V
	\end{equation}
	such that, $D$ is a fully populated clause over $V_{sub}$.	
	Or we can write,
	\begin{equation}
	\label{eq:fpcl1}
	\forall D \in \mathcal{P}(C_i),\; \exists V_{sub} \subseteq V
	\end{equation}
	such that, $D$ is a fully populated clause over $V_{sub}$.	

	Now, We define a set, $E$,
	\begin{equation}
		E=\{y | y=\begin{cases}
					x, \; if \; x \in C_j\\
					\neg x \; if \; \neg x \in C_j\\
					\end{cases} \; and \; x\in V_{sub} \}
	\end{equation}
	$$\implies \forall x \in E \implies x \in C_j$$
	\begin{equation}
		\label{eq:fpcl2}
		\implies E\subseteq C_j
	\end{equation}

	as $C_j$ is a fully populated clause over $V$, and $V_{sub}\subseteq V$, 
	
	$$(\forall x) (x \in V \Leftrightarrow x \in C_j \; or \; \neg x \in C_j)$$
	
	$$\implies \forall x \in V_{sub} \implies x \in C_j \; or \; \neg x \in C_j$$
	
	$$\implies (\forall x)(x \in V_{sub} \Leftrightarrow x \in E \; or \; \neg x \in E)$$
	
	Thus, $E$ is a fully populated clause over $V_{sub}$.
	
	From \cref{eq:fpcl1} and \cref{eq:fpcl2}, 
	\begin{equation}
	\forall D \in \mathcal{P}(C_i),\; \exists E \in \mathcal{P}(C_j)
	\end{equation}
	such that, $D$ and $E$ are fully populated clauses over a common variable set, $V_{sub}$.
	\begin{equation}
	\label{eq:fpcl3}
		\implies (\forall x) (x\in V_{sub} \Leftrightarrow x \in D \; or \; \neg x \in D \Leftrightarrow x \in E \; or \; \neg x \in E)
	\end{equation}
	Now, there can be two cases, either $D=E$ or $D\ne E$,
	
	Suppose, $D=E$
	
	from \cref{eq:fpcl2}, $E\in \mathcal{P}(C_j)$
	$$\implies D \in \mathcal{P}(C_j)$$
	
	But, given that, $D \notin \mathcal{P}(C_j)$
	$$\implies D\ne E$$ 
	Thus, from \cref{eq:fpcl3}, 
	$\implies$ $D$ and $E$ are two unequal fully populated clauses over a common variable set, $V_{sub}$
	
	$\implies$ $D$ and $E$ are sibling clauses.
	
	Hence, 
	$$\forall D \in \mathcal{P}(C_i) \; | \; D\notin \mathcal{P}(C_j)$$
	$$\exists E \in \mathcal{P}(C_j)$$
	such that, $D$ and $E$ are sibling clauses.
\end{proof}

\section{Cardinality of a Boolean formula in CNF}
\label{sec:cCNF}
As we know that, a clause is a set of literals. For a variable $x$, there are two literals, i.e. $x$ and $\neg x$. Let us represent each literal as $l_i$. So, for each variable $x$, there are two literals $l_1$ and $l_2$. For a variable set, $V$, of $n$ variables, there are $2n$ literals. So, a general clause in a boolean formula, $F_{gen}$ can be written in the form, given by,
$$C=\{ l_1, l_2 \dots , l_n, l_{n+1}, l_{n+2}, \dots ,l_{2n} \} $$
where,
$$l_i=\begin{cases}
		l_i\\
		null\\
\end{cases} \; \; \forall l_i \in C
$$
As each $l_i$ can be selected in two ways, independently, and there are $2n$ literals in $C$. So, using fundamental counting principle, the total number of clauses possible are given by, 
$$|F_{gen}|=2 \times 2 \times 2 \times \dots 2n \; times  $$
$$\implies |F_{gen}|=2^{2n}$$

Hence, maximum possible cardinality of a boolean formula in CNF, is $2^{2n}$, including a null clause, $\phi$.
\section{Boolean formula in effective CNF}
A boolean formula in CNF, given by, $F$, is called a boolean formula in effective CNF, if it does not contain a tautology clause. We can write,
\begin{equation}
 \forall C \in F \implies C \; is \; not \; a \; tautology
\end{equation}
Or, $C$ is non-tautology clause. As discussed in \cref{sec:st} , a tautology clause has no effect on satisfiability of a CNF. So, for any given boolean formula, if we can identify tautology clauses, and ignore their existence, we can get an effective CNF.

\section{A complete boolean formula}
A boolean formula, $F_n$, containing every possible non-tautology clause, over a set of variables, $V$, including $null$ clause, is called a complete boolean formula. For eg. for variable set, $V=\{x_1, x_2\}$, the complete boolean formula is given by, 
$$
F_2=(x_1\vee x_2)\wedge
	(x_1\vee \neg x_2)\wedge
	(\neg x_1 \vee x_2)\wedge
	(\neg x_1 \vee \neg x_2)\wedge
	(x_1)\wedge
	(\neg x_1)\wedge
	(x_2)\wedge
	(\neg x_2)\wedge
	\phi
$$
where $\phi$ is a null clause.

In sets notation, it can be written as:
$$
F_2=\{\{x_1,x_2\},\{x_1,\neg x_2\},\{\neg x_1,x_2\},\{\neg x_1,\neg x_2\},\{x_1\},\{\neg x_1\},\{x_2\},\{\neg x_2\},\phi\}
$$

\begin{theorem}
	\label{thm:3n}
	If $F_n$ is a complete boolean formula, over a variable set, $V$, of $n$ variables, then $F_n$ contains $3^n$ clauses, including a $null$ clause.
\end{theorem}
\begin{proof}
	Given, that $F_n$ is a complete boolean formula, over $V$, and $V$ contains $n$ variables.
	$$\implies |V|=n$$
	
	From the definition of a complete boolean formula, we know that, 
	$$\forall C \in F_n \implies C \; is \; a \; non-tautology \; clause$$
	
	$$\implies \forall x \in C \implies \neg x \notin C$$
	We can write, a general clause in $F_n$ as,
	\begin{equation}
	C=(X_1,X_2,X_3,\dots X_n)
	\end{equation}
	where, $$
	X_i=\begin{cases}
	x_i \\
	\neg x_i \\
	null \\
	\end{cases}
	$$
	$X_i$ is a variable, which can be assigned the values $x_i, \neg x_i$ or $null$ independently. A $null$ value for $X_i$ means, the clause $C$, neither contain $x_i$ nor $\neg x_i$.
	
	Now, Each $X_i$ can be assigned a value in 3 different ways. By using fundamental counting principle, the total number of clauses possible is given by, 
	$$
	n(C)=3\times 3 \times 3 \dots n \; times
	\implies n(C)=3^n
	$$
	Also, there will be a clause in which, $X_i$ is assigned $null$ value $\forall X_i \in C$. It will result in a $null$ clause, given by, $\phi$.
	Hence, $F_n$ contains $3^n$ clauses, including a $null$ clause.
\end{proof}

\begin{corollary}
	\label{cor:3n-1}
	If $F_n$ is a complete boolean formula, over a variable set, $V$, with $n$ variables, then, for any given variable $x_i$,
	$$
	n(x_i)=n(\neg x_i)=n(x_{i-null})=3^{n-1}
	$$
	where,
	$n(x_i)$ is number of clauses containing $x_i$,
	
	$n(\neg x_i)$ is number of clauses containing $\neg x_i$
	
	$n(x_{i-null})$ is number of clauses containing neither $x_i$ nor $\neg x_i$
	
\end{corollary}
\begin{proof}
	As explained in \cref{thm:3n}, for the given complete boolean formula $F_n$, with n clauses, we can write a clause in the form, given by,
	\begin{equation}
	\label{eq:c2}
	C=(X_1,X_2,X_3,\dots X_n)
	\end{equation}
	where,
	$$
	X_i=\begin{cases}
	x_i \\
	\neg x_i \\
	null \\
	\end{cases}
	$$
	Now, suppose, we put $X_i=x_i$ for some $i\in [1,n]$ in \cref{eq:c2}, we get,
	$$
	C=(X_1,X_2,X_3,\dots x_i, \dots X_n)
	$$
	$$
	C=(x_i, X_1,X_2,X_3,\dots X_n)
	$$
	We have assigned the value $x_i$ to one of the variables. There are $n-1$ variables remaining, to which, we can assign values independently. Each variable can be assigned 3 values independently. Thus, by using the fundamental counting principle, the total number of clauses with $x_i$ is given by,
	$$
	n(x_i)=3\times 3 \times 3 \dots n-1 \; times
	$$
	$$
	n(x_i)=3^{n-1}
	$$
	Similarily, by assigning $X_i=\neg x_i$ and $X_i=null$ we find 
	$$
	n(\neg x_i)=3^{n-1}
	$$
	and 
	$$
	n(x_{i-null})=3^{n-1}
	$$
	Hence, we get,
	$$
	n(x_i)=n(\neg x_i)=n(x_{i-null})=3^{n-1}
	$$
\end{proof}

\begin{corollary}
	\label{co:cbf1}
	A complete boolean formula, $F_n$ can be written as:
	$$	F_n=\mathcal{P}(C_1) \cup\mathcal{P}(C_2) \cup \dots \cup \mathcal{P}(C_p)$$
	 where $\{C_1,C_2,\dots C_p \}$is set of all poosible fully populated clauses over a set of variables, $V$.
\end{corollary} 
\begin{proof}
	As explained in \cref{thm:3n}, for the given complete boolean formula $F_n$, with n clauses, we can write a clause in the form, given by,
	\begin{equation}
	\label{eq:c3}
	C=(X_1,X_2,X_3,\dots X_n)
	\end{equation}
	where,
	$$
	X_i=\begin{cases}
	x_i \\
	\neg x_i \\
	null \\
	\end{cases}
	$$
	But, First, if we assign 
	$$
	X_i=\begin{cases}
	x_i \\
	\neg x_i \\
	\end{cases} \forall X_i \in C
	$$
	We get a set of all fully populated clauses over $V$, say $F_{full}$, given by,
	$$F_{full}=\{C_1,C_2,\dots, C_p \}$$
	
	Then, we assign, for any $C_i \in F_{full}$
	$$
	X_i=\begin{cases}
	x \; | \; x \in C_i \\
	null\\
	\end{cases} 
	$$
	We get power set of clause $C_i$. By assigning vaues, as above, $\forall C_i \in F_{full}$, we get all possible clauses over $V$. Thus, we can write:
	\begin{equation}
	\label{eq:cf}
	F_n=\mathcal{P}(C_1) \cup\mathcal{P}(C_2) \cup \dots \cup \mathcal{P}(C_p)
	\end{equation}
	Hence proved.
\end{proof}

\begin{theorem}
	\label{thm:2n-1}
	If $\mathcal{P}(C)$ is a powerset of C, where C is a fully populated clause, over a variable set, $V$, with $n$ variables, then, 
	$$
	n(x_i)=n(x_{i-null})=2^{n-1} \; \; \forall x_i \in C
	$$
	where,
	$n(x_i)$ is number of clauses containing $x_i$, in $\mathcal{P}(C)$ and \\
	$n(x_{i-null})$ is number of clauses not containing $x_i$, in $\mathcal{P}(C)$
	
\end{theorem}
\begin{proof}
	Given that, $C$ is a fully populated clause,  over a variable set, $V$, with $n$ variables, and $\mathcal{P}(C)$ is a power set of $C$.
	Let $D\in \mathcal{P}(C)$. We can write, $D$, in general form,
	$$
	D=(X_1,X_2,X_3 \dots X_n) 
	$$ 
	where,
	$$
	X_i=\begin{cases}
	x_i & x_i \in C\\
	null \\
	\end{cases}
	$$
	Now, if we put $X_i=x_i$ for some $i\in [1,n]$, we get,
	$$
	D=(X_1,X_2,X_3 \dots x_i \dots X_n)
	$$
	$$
	D=(x_i, X_1,X_2,X_3 \dots X_n)
	$$
	We have assigned the value $x_i$ to one of the variables. There are $n-1$ variables re-
	maining, to which, we can assign values independently. Each variable can be assigned
	2 values independently. Thus, by using the fundamental counting principle, the total
	number of clauses with $x_i$ is given by,
	
	$$
	n(x_i)=2 \times2\times2\dots (n-1) times
	$$
	$$
	n(x_i)=2^{n-1}
	$$
	Similarily, by assigning $X_i=null$ we find  
	$$
	n(x_{i-null})=2^{n-1}
	$$
	Hence, we get,
	$$
	n(x_i)=n(x_{i-null})=2^{n-1}
	$$
\end{proof}

\begin{theorem}
		\label{thm:SatF}

		If there exists a fully populated clause, $C_k$ over $V$, such that, 
		$$F=F_n \setminus \mathcal{P}(C_k)$$
		where, $F_n$ is a complete boolean formula over $V$, then, $F$ is satisfiable.
\end{theorem}
\begin{proof}
		Given that, 
		$$F=F_n \setminus \mathcal{P}(C_k)$$
		
		Suppose $D$ is any clause in $F$, i.e. $D\in F$
		
		$$\implies D \in F_n \setminus \mathcal{P}(C_k)$$
		
		$$\implies D \in F_n \; | \; D \notin \mathcal{P}(C_k)$$
		
		Let $F_{full}=\{C_1, C_2, \dots , C_p\}$ is a set of all fully populated clauses over $V$, then from \cref{co:cbf1}
		$$\implies D \in \mathcal{P}(C_1)\cup \mathcal{P}(C_2)\cup \dots \cup \mathcal{P}(C_p) \; | \; D \notin \mathcal{P}(C_k)$$
		
		As, $C_1, C_2, \dots , C_p$ are unequal fully populated clauses over $V$, which implies, from the definition of sibling clauses, $C_1, C_2, \dots , C_p$ are sibling clauses, including $C_k$.
		
		Let $D \in \mathcal{P}(C_i)$, where $C_i$ is any clause in $F_{full}$ but $C_i \ne C_k$
		$$\implies D \in \mathcal{P}(C_i) \; | \; D \notin \mathcal{P}(C_k)$$
		from \cref{lemma:desibling} 
		$$\exists E \in \mathcal{P}(C_k)$$
		such that, $D$ and $E$ are sibling clauses.

		As $D$ is any clause in $F$
		$$\implies \forall D \in F, \; \exists E \in \mathcal{P}(C_k)$$
		such that, $D$ and $E$ are sibling clauses.
		
		As, $C_k$ is a fully populated clause, so for a valuation, given by,
		$$C_k=0(false)$$
		from \cref{thm:c0d0}
		$$\forall E \subseteq C_k \implies E=0(false)$$
		$$\implies \forall E \in \mathcal{P}(C_k) \implies E=0(false)$$
		As $D$ and $E$ are sibling clauses, from \cref{lemma:a0b1}
		$$\implies \forall D \in F \implies D=1(true)$$
		for valuation $C_k=0(false)$
		
		$\implies$ $F$ is  satisfiable.
		
\end{proof}

\begin{theorem}
	\label{thm:SubF}
	If $F$ is satisfiable, and $F_{sub} \subseteq F$, then $F_{sub}$ is satisfiable.
\end{theorem}
\begin{proof}
	Given that, 
	$$F_{sub} \subseteq F$$
	$$\implies \forall C \in F_{sub} \implies C \in F$$
	As $F$ is satisfiable. It implies, there exists a valuation, such that, $\forall D \in F \implies D=1(true)$
	$$\implies \forall C \in F_{sub} \implies C=1(true)$$
	Hence, $F_{sub}$ is satisfiable.
\end{proof}

\begin{theorem}
		\label{thm:FSat}
		If $F$ is satisfiable, then, there exists a fully populated clause, $C_k$, such that,
		$$\forall E \in \mathcal{P}(C_k) \implies E \notin F$$
\end{theorem}
\begin{proof}
		Given that, $F$ is satisfiable.

		Now, suppose, for the sake of contradiction, that, there does not exist a fully populated clause, $C_k$, such that,
		$$\forall E \in \mathcal{P}(C_k) \implies E \notin F$$
		
		$$\implies \forall C_k, \; \exists E \in \mathcal{P}(C_k) | E\in F$$
		
		From \cref{lemma:ck0}, for any valuation, to the variable set, $V$, $\exists C_k | C_k=0(false)$
		
		$\implies$ for any valuation, to the variable set, $V$, in which, $C_k=0(false)$
		
		$$\exists E \in \mathcal{P}(C_k) | E \in F$$
		
		$E \in \mathcal{P}(C_k) \implies E \subseteq C_k$, from \cref{thm:c0d0}, for any valuation, to the variable set, $V$, in which, $C_k=0(false) \implies E=0(false)$
		
		$\implies$ for any valuation, to the variable set, $V$, 
		$$\exists E \in F | E=0$$
		
		$\implies$ $F$ is unsatisfiable. Which is a contradiction, so our assumption was wrong. Hence, there exist a fully populated clause, $C_k$, such that,
		$$\forall E \in \mathcal{P}(C_k) \implies E \notin F $$
		Hence proved.
		 
\end{proof}
\begin{theorem}[Necessaery and Sufficient Condition for Satisfiability]
		\label{thm:satiff}
		$F$ is satisfiable, if and only if, there exists a fully populated clause, $C_k$, such that,  
		$$\forall E \in \mathcal{P}(C_k) \implies E \notin F$$
\end{theorem}

\begin{proof}
	Let $F$ is satisfiable, from \cref{thm:FSat}, it implies,
	\begin{equation}
		\label{eq:satiff1}
		\forall E \in \mathcal{P}(C_k) \implies E \notin F 
	\end{equation}
	
	Conversely, Let there exists a fully populated clause, $C_k$, such that,
	
	$$\forall E \in \mathcal{P}(C_k) \implies E \notin F $$
	
	Let $F_n$ is a complete boolean formula, over $V$, then we can write,
	$$F\subseteq F_n \setminus \mathcal{P}(C_k)$$
	
	From \cref{thm:SatF} and \cref{thm:SubF}, we get, $F$  is satisfiable.

	 Hence, $F$ is satisfiable, if and only if, there exists a fully populated clause, $C_k$, such that,  
	 $$\forall E \in \mathcal{P}(C_k) \implies E \notin F$$
\end{proof}

\section{Algorithm to solve boolean satisfiability problem}
We know that, a given boolean formula,$F$ is satisfiable, if and only if, there exists a fully populated clause, $C_k$, such that, 
$$\forall E \in \mathcal{P}(C_k) \implies E \notin F$$
Now, in order to find such clause, $C_k$, we must know set of variables, $V$, over which boolean formula, $F$, has been defined. So, the first step of the algorithm will be to find the variable set, $V$, then generate set of fully populated clauses, $S$, over $V$. In next step, for each clause in $F$, we eliminate their superset fully populated clauses from $S$. After processing all clauses in $F$, $S$ contains fully populated clauses, which (or their subsets) are absent in $F$. Now, for each remaining fully populated clause, $C_{absent}$, we can return a solution, $S_i$ such that,

$$
S_i=\{ x | x= \neg y \forall y \in C_{absent} \}
$$
or simply, we can print 
$F$ is satisfiable for, $C_{absent}=0$

However, if no fully populated clause remained after processing all clauses, then $F$ is unsatisfiable, from \cref{thm:satiff}

\subsection{Representation of fully populated clauses}
Fully populated clauses can be represented in the form of a binary tree, as shown in \cref{imgfpctree}, in which a node contains the index value, $i$, of a variable, $x_i$. Non-null value in "left" pointer represents literal $\neg x_i$ and non-null value in "right" pointer represents literal $x_i$. 
 Each edge, $E$ of an $i^{th}$ node represents a clause over $V$, where $V=\{ x_1,x_2, \dots x_i \}$. 
 
 "left" and "right" pointers of a node can have following values: 
 \begin{itemize}
 	\item "OPEN": Initial value, represents insertion point for a new variable.
 	\item Address of child node.
 	\item "NULL": Null value represents presence of subset of fully populated clauses in $F$, hence exclusion from solution.
 \end{itemize}
 
 \begin{figure}[H]
	\label{imgfpctree}
	\caption{Representation of fully populated clauses over  $\{ x_1,x_2 \}$ as a binary tree}
	\centering
	\includegraphics[height=200pt,width=300pt]{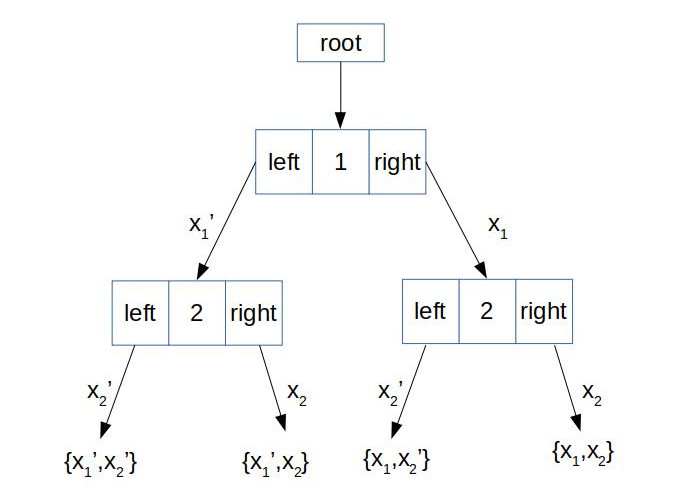}
\end{figure}

\subsection{Elimination of superset fully populated clauses}
For a given clause $C_i$ in $F$, we find edge, $E$ in the tree, where, $x_i=y_i \forall x_i \in C_i$, where,  $ y_i \in E  $, and remove that edge, by putting source pointer to NULL value. For e.g. for $F=\{\{ \neg x_1\}, \{x_1,\neg x_2 \}\}$ the set of fully populated clauses, shown in \cref{imgfpctree} shall go under transformations, shown in \cref{imgprocess}

\begin{figure}[h]
	\label{imgprocess}
	\caption{Elimination of superset fully populated clauses for $F=\{\{ \neg x_1\}, \{x_1,\neg x_2 \}\}$}
	\centering
	\includegraphics[height=250pt,width=350pt]{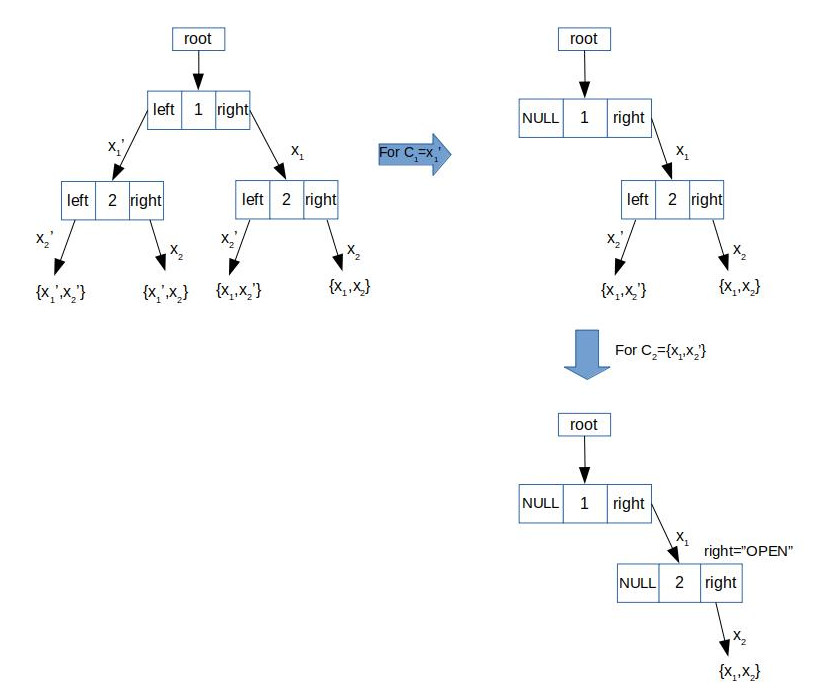}
\end{figure}

After processing all clauses in $F$, we get a set of fully populated clauses, $S$ which (and their subsets) are absent in $F$, hence, for each $S_i \in S$, $F$ is satisfiable for $S_i=0(false)$. For e.g. in \cref{imgprocess}, we are left with a fully populated clause, $S_1=\{x_1,x_2 \}$ for $F=\{\{ \neg x_1\}, \{x_1,\neg x_2 \}\}$ , which implies, $F$ is satisfiable for $S_1=\{x_1,x_2 \}=0(false)$, which is true. Further, $F$ is unsatisfiable for any other truth assignment.
\subsection{Optimisations to the algorithm}
The steps explained above are required to find the solution to any given boolean formula. However, different types of formulas are possible for a given set of variables, so different kinds of optimisations can be applied to the algorithm to solve a particular type of formula.

A clause with lower cardinality is subset of more fully populated clauses, so, the given formula can be sorted in ascending order of cardinality of clauses. It will enable the algorithm to eliminate larger number of fully populated clauses at initial stage, which shall result in effective memory management.

A null clause, $\phi$ is subset of all fully populated clauses. So, if $F$ contains $\phi$, it implies, no fully populated clause is absent in the formula along with it's all subsets. Which implies, $F$ is unsatisfiable. So, if algorithm find $\phi$ in $F$, it shall return "unsatisfiable formula" and exit.

If the algorithm generates set of fully populated clauses for one clause at a time, it will result in effective memory management. If used after sorting, it will reduce the space and time required to produce results. 
\begin{algorithm}[H]
	\label{alg:checkSAT}
	\caption{CheckSAT(F)}
	\begin{algorithmic}
		\STATE{Define struct node(left,i,right)}
		\STATE{Define const OPEN=node(NULL,-1,NULL)}
		\STATE{Define node *root}
		\STATE{Define list V,E,C}
		\STATE{Define boolean SAT=false, closedTree=true, tautologyClause=false}
		\STATE{INPUT list F}
		\STATE{SORT F}
		\STATE{SET root=OPEN}
		\FOR{each $C_i \in F$}
			\IF{$C_i=\phi$}
				\STATE{Set SAT:=false}
				\RETURN SAT
			\ENDIF
			\STATE{Set tautologyClause=false}
			\FOR{each $x_j \in C_i$}
				\IF{$\neg x_j \in C_i$}
					\STATE{Set tautologyClause=true}
					\STATE{Break loop}
				\ENDIF
				
				\IF{$x_j \notin V$}
					\STATE{INSERT $x_j$ in $V$}
					\STATE{Set closedTree:=true}
					\FOR{each pointer in TREE(root)}
						\IF{pointer=OPEN}
							\STATE{Insert node(OPEN,j,OPEN)}
							\STATE{Set closedTree:=false}
						\ENDIF
					\ENDFOR
					\IF{closedTree=true}
						\STATE{Set SAT:=false}
						\RETURN SAT
					\ENDIF
				\ENDIF
			\ENDFOR
			\FOR{each pointer in TREE(root)}
				\IF{tautologyClause=true}
					\STATE{Break loop}
				\ENDIF
				\IF{$C_i\subseteq E_{pointer}$}
					\STATE{Set pointer:=NULL}
				\ENDIF
			\ENDFOR
		\ENDFOR
		\FOR{each pointer in TREE(root)}
			\IF{pointer=OPEN}
				\PRINT{ $F$ is satisfiable for $E_{pointer}=0$}
				\STATE{Set SAT:=true}
			\ENDIF
		\ENDFOR 
		\RETURN SAT
	\end{algorithmic}
\end{algorithm}

\subsection{Illustration}
Let, 
$$F=\{\{\neg x_1,\neg x_2\}, \{x_3\},\{\neg x_1\},\{x_1,\neg x_2,\neg x_3 \}\}$$
after sorting $F$, we get,
$$F=\{ \{x_3\}, \{\neg x_1\},\{\neg x_1,\neg x_2\},\{x_1,\neg x_2,\neg x_3 \}\}$$

for $C_1=\{x_3\}$, we add node to root, to get a tree as shown in \cref{imgi1}

\begin{figure}[H]
	\label{imgi1}
	\caption{Tree after adding node for  $C_1=\{x_3\}$}
	\centering
	\includegraphics[height=100pt,width=120pt]{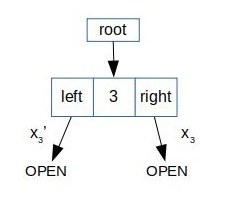}
\end{figure}

Now, we put $ptr=NULL$, where $C_1 \subseteq E_{ptr}$, for  $C_1=\{x_3\}$, we get,

\begin{figure}[H]
	\label{imgi2}
	\caption{Tree after putting $ptr=NULL$, where $C_1 \subseteq E_{ptr}$}
	\centering
	\includegraphics[height=100pt,width=100pt]{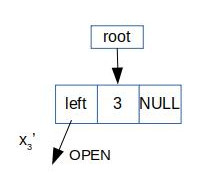}
\end{figure}

Now, for $C_2=\{\neg x_1\}$, we add node to ptr, where ptr=OPEN, to get a tree as shown in \cref{imgi3}
\begin{figure}[H]
	\label{imgi3}
	\caption{Tree after adding node for  $C_2=\{\neg x_1\}$}
	\centering
	\includegraphics[height=150pt,width=150pt]{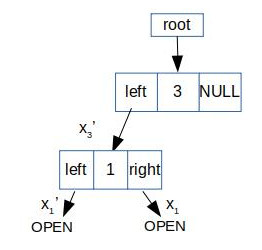}
\end{figure}

Now, we put $ptr=NULL$, where $C_2 \subseteq E_{ptr}$, for $C_2=\{\neg x_1\}$, we get,

\begin{figure}[H]
	\label{imgi4}
	\caption{Tree after putting $ptr=NULL$, where $C_2 \subseteq E_{ptr}$}
	\centering
	\includegraphics[height=150pt,width=150pt]{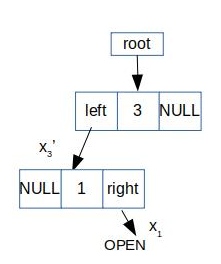}
\end{figure}

Now, for $C_3=\{\neg x_1,\neg x_2\}$, we add node to ptr, where ptr=OPEN, to get a tree as shown in \cref{imgi5}
\begin{figure}[H]
	\label{imgi5}
	\caption{Tree after adding nodes for  $C_3=\{\neg x_1,\neg x_2\}$}
	\centering
	\includegraphics[height=200pt,width=150pt]{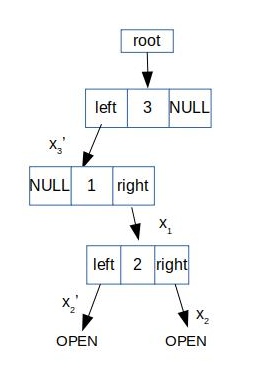}
\end{figure}

Now, we put $ptr=NULL$, where $C_3 \subseteq E_{ptr}$, for $C_3=\{\neg x_1,\neg x_2\}$, we get,

\begin{figure}[H]
	\label{imgi6}
	\caption{Tree after putting $ptr=NULL$, where $C_3 \subseteq E_{ptr}$}
	\centering
	\includegraphics[height=200pt,width=150pt]{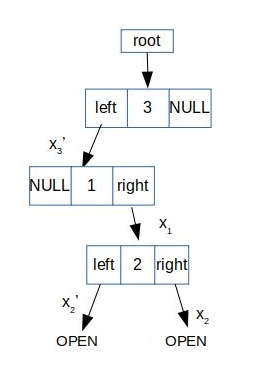}
\end{figure}
But, tree already have eliminated fully populated clauses, which are supersets of $\{\neg x_1\}$, in previous steps, so it remained unchanged.

In next step, for $C_4=\{ x_1,\neg x_2,\neg x_3\}$, $C_4$ does not contains a new node, so tree remains the same. 

Now, we put $ptr=NULL$, where $C_4 \subseteq E_{ptr}$, for $C_4=\{ x_1,\neg x_2,\neg x_3\}$, we get,

\begin{figure}[H]
	\label{imgi7}
	\caption{Tree after putting $ptr=NULL$, where $C_4 \subseteq E_{ptr}$}
	\centering
	\includegraphics[height=200pt,width=150pt]{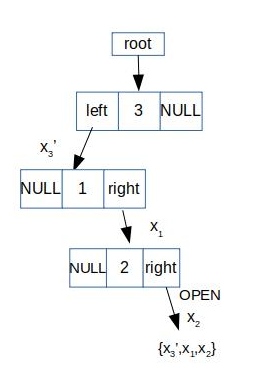}
\end{figure}

After processing all clauses, algorithm prints the solution, i.e. "$F$ is satisfiable for $\{\neg x_3,x_1,x_2\}=0$", where, $$F=\{\{\neg x_1,\neg x_2\}, \{x_3\},\{\neg x_1\},\{x_1,\neg x_2,\neg x_3 \}\}$$

\section{Implications on P vs NP Problem}
The above mentioned algorithm runs in polynomial time on it's input size, which implies,
$$SATISFIABILITY \in P$$

From Corollary 1. in \cite{Karp1972}
$$\implies P = NP$$

\section{Cardinality Function}
Cardinality function can be used to determine cardinality of a given boolean formula. It can be used to optimise existing algorithms. For a given boolean formula in CNF, we define a function $f(X,X_c)$ where, $X=(x_1,x_2,x_3,...,x_n)$ and $X_c=(\neg x_1,\neg x_2,\neg x_3,...,\neg x_n)$, by replacing disjunction with multiplication and conjuction with addition. For e.g.
Let the boolean formula, $F$, is given by,

\begin{equation}
\label{eq:p_f0}
F=(x_1 \vee x_2 \vee \neg x_3)\wedge(\neg x_1 \vee x_2 \vee x_3)
\end{equation}

We define the function, $f(X,X_c)$, given by:
\begin{equation}
f(X,X_c)=(x_1\*x_2\*\neg x_3)+(\neg x_1 \* x_2 \*x_3)
\end{equation}

In general, $f(X,X_c)$ can be defined as: 
\begin{equation}
\label{eq:p_f1}
f(X,X_c)=\sum_{i=1}^{p}M_i \; \; \forall i \in [1,p] 
\end{equation}
where, $p$ is number of clauses in $F$ and 
\begin{equation}
\label{eq:p_f2}
M_i=\prod x_j \; \; \forall x_j\in C_i
\end{equation}
where, $C_i$ is $i^{th}$ clause in the given boolean formula, $F$. It is to be noted that $f(X,X_c)$ is a function on integers, i.e. $f:\mathbf{Z} \rightarrow \mathbf{Z}$

The algorithm for $f(X,XC)$ for the boolean formula given in \cref{eq:p_f0} is given below:

\begin{algorithm}[H]
	\caption{$f(X,XC)$}
	\label{alg:p_f}
	\begin{algorithmic}
		\RETURN $(X[1]*X[2]*XC[3])+(XC[1] * X[2] *X[3] )$
	\end{algorithmic}
\end{algorithm}

\subsection{Total number of clauses}
Following algorithm can be used to check total number of clauses in $F$
\begin{algorithm}[H]
	\caption{TotalClauses()}
	\label{alg:total}
	\begin{algorithmic}
		\STATE{Set X:=1, XC=1}
		\RETURN f(X,XC)
	\end{algorithmic}
\end{algorithm}

\subsection{Total number of clauses containing $x_i$}
Following algorithm can be used to check total number of clauses, containing $x_i$ in $F$
\begin{algorithm}[H]
	\caption{n($i$)}
	\label{alg:nxi}
	\begin{algorithmic}
		\STATE{Set T:=$TotalClauses()$}
		\STATE{Set X[i]=0}
		\RETURN T-f(X,XC)
	\end{algorithmic}
\end{algorithm}

\subsection{Total number of clauses containing $\neg x_i$}
Following algorithm can be used to check total number of clauses, containing $\neg x_i$ in $F$
\begin{algorithm}[H]
	\caption{nc($i$)}
	\label{alg:ncxci}
	\begin{algorithmic}
		\STATE{Set T:=$TotalClauses()$}
		\STATE{Set XC[i]=0}
		\RETURN T-f(X,XC)
	\end{algorithmic}
\end{algorithm}

\subsection{Total number of clauses containing $x_i$ or $\neg x_i$}
Following algorithm can be used to check total number of clauses, containing $x_i$ or $\neg x_i$ in $F$
\begin{algorithm}[H]
	\caption{nxUxi($i$)}
	\label{alg:nii}
	\begin{algorithmic}
		\STATE{Set T:=$TotalClauses()$}
		\STATE{Set X[i]:=0, XC[i]=0}
		\RETURN T-f(X,XC)
	\end{algorithmic}
\end{algorithm}

\subsection{Checking tautology clauses in $F$}
The following algorithm can be used to check existence of tautology clauses, in a given formula, $F$, in polynomial time.

\begin{algorithm}[H]
	\label{alg:checkTaut}
	\caption{CheckTautologyClauses()}
	\label{alg:chk_taut}
	\begin{algorithmic}
		\STATE{Define i:=0}
		\STATE{Set i:=1} 
		\WHILE{$i\le n$}
			\IF{$n(i)+nc(i)> nxUxi(i)$} 
				\RETURN $true$
			\ENDIF
			\STATE{Update $i:=i+1$}
		\ENDWHILE
		\RETURN $false$
	\end{algorithmic}
\end{algorithm}

\section{Optimisations based on cardinality}
These optimisations can also be integrated with algorithm given in \cref{alg:checkSAT}.

Suppose, $F_n$ is a complete boolean formula over a variable set, $V$, and $C_k$ is a fully populated clause over $V$, and $F$ is a boolean formula, such that,
$$F=F_n \setminus \mathcal{P}(C_k)$$

From \cref{thm:SatF}, $\implies$ $F$  is satisfiable. 

And, $$|F|= |F_n|-|\mathcal{P}(C_k)|$$ 

We know that cardinality of a power set of $C_k$ is $2^n$, where $|C_k|=n$,
i.e.
\begin{equation}
\label{eq:opt:2n}
 |\mathcal{P}(C_k)|=2^n
\end{equation}

from \cref{thm:3n}, $|F_n|=3^n$ 
$$\implies |F|=3^n-2^n$$

and from \cref{cor:3n-1} and \cref{thm:2n-1} $$\forall x_i \in C_k \implies n(x_i)=3^{n-1}-2^{n-1}$$
and 
$$\forall \neg x_i \in C_k \implies n(\neg x_i)=3^{n-1}-2^{n-1}$$
where, $n( x_i)$ is number of clauses containing $ x_i$ in $F$, and
 $n(\neg x_i)$ is number of clauses containing $\neg x_i$ in $F$.

Now, Suppose $D$ is any non-tautology clause over variable set $V$ such that, $D \notin F$,
$$\implies D \notin (F_n \setminus \mathcal{P}(C_k))$$ 
But, $F_n$ is a complete boolean formula over $V$,
$$\implies D \in \mathcal{P}(C_k)$$
Now, we define a formula $F_{new}$, such that,
$$F_{new}=F\wedge D$$
It implies, there does not exists a fully populated clause, $C_k$, such that,
$$\forall D \in \mathcal{P}(C_k) \implies D \notin F_{new}$$
From \cref{thm:satiff},
$\implies$ $F_{new}$ is unsatisfiable.

It implies, inclusion of any non-tautology clause in $F$, results in unsatisfiable formula.

So, for a boolean formula $F$, in effective CNF, which can be checked using algorithm given in \cref{alg:checkTaut}, over a variable set $V$, of $n$ variables, we have, 
\begin{itemize}
	\item If $|F|>3^n-2^n$ then $F$ is unsatisfiable.
	\item If $\exists x_i \in V \; | \; min(n(x_i),n(\neg x_i))>3^{n-1}-2^{n-1}$ then $F$ is unsatisfiable.
	\item If $\exists x_i \in V \; | \; n(x_i)\le 3^{n-1}-2^{n-1}<n(\neg x_i)$ then $x_i=0$, $\neg x_i=1$ belong to the solution, if any.
	\item If $\exists x_i \in V \; | \; n(\neg x_i)\le 3^{n-1}-2^{n-1}<n( x_i)$ then $\neg x_i=0$, $x_i=1$ belong to the solution, if any.
\end{itemize}

Suppose, the given boolean formula, $F_{gen}$,  is not in effective CNF, which can be checked using algorithm given in \cref{alg:checkTaut}, then, from \cref{sec:cCNF}, we know that maximum possible cardinality of a general boolean formula in CNF, is $2^{2n}$, including a null clause, $\phi$. Which implies,
\begin{itemize}
		\item If $|F_{gen}|>2^{2n}-2^n$ then $F_{gen}$ is unsatisfiable.
\end{itemize}

\section{Conclusion}
\label{sec:conclusion}
A necessary and sufficient condition has been established to determine satisfiability of a boolean formula in CNF, which has been used to develop a polynomial time algorithm to solve boolean satisfiability problem. Existence of a polynomial time algorithm to solve boolean satisfiability problem implies $P=NP$.

\bibliographystyle{siamplain}
\bibliography{pvsnparxiv}
\end{document}